\documentclass[11pt,letterpaper]{article}

\usepackage[margin=1in]{geometry}

\usepackage{amsmath,amsfonts,amssymb,amsthm}

\usepackage{tikz}
\usetikzlibrary{arrows.meta,calc}

\usepackage[ruled,vlined,linesnumbered]{algorithm2e}

\usepackage{enumitem}
\usepackage[numbers,sort]{natbib}

\usepackage{xcolor}
\usepackage[
    colorlinks=true,
    linkcolor=blue,
    citecolor=blue,
    urlcolor=blue
]{hyperref}

\theoremstyle{plain}
\newtheorem{theorem}{Theorem}[section]
\newtheorem{lemma}[theorem]{Lemma}

\theoremstyle{definition}
\newtheorem{definition}[theorem]{Definition}

\theoremstyle{remark}
\newtheorem{remark}[theorem]{Remark}

\newcommand{\fav}[2][]{\operatorname{top}\if\relax\detokenize{#1}\relax\else_{#1}\fi\!\left(#2\right)}
\newcommand{\cost}[2][]{\operatorname{cost}\if\relax\detokenize{#1}\relax\else_{#1}\fi\!\left(#2\right)}
\newcommand{\eps}{\epsilon}

\newcommand{\profile}{\overrightarrow \succ}

\DeclareRobustCommand{\rulename}[1]{\textnormal{\textsc{#1}}}
\DeclareRobustCommand{\protocolname}{\rulename{TreePluralityVeto}}
\DeclareRobustCommand{\protocolshort}{\rulename{TreePV}}

\title{Tree Plurality Veto: Voting in Few Rounds\\ with Efficient Communication and Low Metric Distortion}
\author{
Gabriel Chen\thanks{Brown University. Email: \texttt{gabriel\_chen@brown.edu}}
\and
Yu Cheng\thanks{Brown University. Email: \texttt{yu\_cheng@brown.edu}}
\and
Alice Wang\thanks{Brown University. Email: \texttt{alice\_z\_wang@brown.edu}}
}
\date{}

\begin{document}

\maketitle

\begin{abstract}
We consider the metric distortion framework for voting, where $n$ voters and $m$ candidates are embedded in an unknown metric space.
A voting rule selects a candidate using only voters' ordinal preferences, and its distortion is the worst-case ratio between the social cost of the selected candidate and the minimum social cost.
We study how round complexity trades off against total communication in deterministic voting protocols with low distortion.
Deterministic rules can achieve the optimal distortion of $3$ in one round by eliciting complete rankings using $\Theta(nm\log m)$ bits, while a direct sequential implementation of \rulename{PluralityVeto} achieves the same distortion using only $O(n\log m)$ bits but requires $\Theta(n)$ rounds.

Our main result is \protocolname{}, a deterministic voting protocol that, for every $r\ge 2$, achieves distortion at most $2r-1$ using at most $r$ rounds and $O\!\left(n\log m+n s^{1/(r-1)}\log s\right)$ bits of total communication, where $s = \min\{n, m\}$.
The protocol runs \rulename{PluralityVeto} on a hierarchical partition of the voters, collecting restricted rankings in parallel at each level.
The branching factor controls the tradeoff between rounds and communication, and a matching-based analysis shows that each level increases the distortion bound by at most $2$.
\end{abstract}

\section{Introduction}

To select a socially good candidate, how much information must voters communicate, and how many rounds of interaction are needed?
We study this question in the metric distortion framework for voting~\cite{AnshelevichBhardwajPostl2015,
AnshelevichBhardwajElkindPostlSkowron2018}, where there are $n$ voters and $m$ candidates in an unknown metric space,
and each voter ranks candidates from nearest to farthest.
The social cost of a candidate is its total distance to the voters, and we evaluate a voting rule by its \emph{distortion}, the worst-case ratio between the social cost of its winner and that of the optimal candidate. Distortion measures the loss from observing only ordinal rankings rather than the underlying cardinal information.

Anshelevich et al.~\cite{AnshelevichBhardwajPostl2015} proved that every deterministic voting rule has distortion at least $3$.
Gkatzelis et al.~\cite{GkatzelisHalpernShah2020} introduced \rulename{PluralityMatching}, the first deterministic voting rule with distortion $3$, and Kizilkaya and Kempe~\cite{KizilkayaKempe2022} later introduced the simpler rule \rulename{PluralityVeto}.
These results settled the optimal distortion of deterministic rules at $3$.
We study the tradeoff among distortion, round complexity, and total communication for deterministic protocols.

Two straightforward protocols illustrate the tradeoff between total communication and the number of rounds.
Eliciting complete rankings requires one round and $\Theta(nm\log m)$ bits of communication, and allows any voting rule to be computed locally.
Alternatively, \rulename{PluralityVeto} can be implemented using only $O(n\log m)$ bits: each voter communicates a top choice and a veto.
However, its direct sequential implementation requires $\Theta(n)$ rounds, because each voter vetoes a candidate based on those remaining after earlier vetoes.

Kempe~\cite{Kempe2020Communication} proved an $\Omega(m/b)$ lower bound on distortion for deterministic one-round mechanisms that use the same $b$-bit encoding for all voters.
This lower bound does not directly apply when message functions may depend on voter identity.
In Appendix~\ref{apx:lower-bound}, we show that any deterministic protocol with distortion at most $\alpha$ requires $\Omega(n/\alpha)$ bits of total communication, even with only two candidates and an unlimited number of rounds.
For constant $\alpha$, this lower bound is tight up to an $O(\log m)$ factor, since \rulename{PluralityVeto} uses $O(n\log m)$ bits.
This motivates our central question:

\begin{quote}
\em
What tradeoff between communication and round complexity is possible for deterministic voting protocols with low distortion?
\end{quote}

Our main result is \protocolname{}, a deterministic voting protocol that provides a tradeoff among distortion, round complexity, and total communication.

\begin{theorem}
\label{thm:main}
Let $n$ be the number of voters, $m$ the number of candidates, and $s = \min\{n, m\}$.
For every integer $r\ge 2$, there is a deterministic voting protocol that achieves distortion at most $2r-1$ using at most $r$ rounds and $O\!\left(n\log m+n s^{1/(r-1)}\log s\right)$ bits of total communication.
\end{theorem}

In the first round, every voter reports their global favorite candidate.
Voters with the same global favorite form a group, and these groups are placed at the leaves of a tree whose branching factor depends on the round budget.
There are at most $s = \min\{n, m\}$ such groups.
Each internal node runs \rulename{PluralityVeto} among the winners
of its children, by asking all voters in its subtree to report their rankings over this restricted set of candidates.
Sub-elections at nodes of the same height run in parallel: all required rankings are collected in one round, after which their winners are computed locally from the public transcript.
A larger branching factor gives fewer levels and hence fewer rounds, but each voter needs to rank more candidates per level.

The analysis uses the perfect-matching certificate of \rulename{PluralityVeto} to compare the winner of a sub-election with any candidate, including those absent from it.
This comparison also appears in the analysis of recursive Fractional Veto by Kizilkaya, Aggarwal, and Micha~\cite{KizilkayaAggarwalMicha2026}.
We apply this comparison across a hierarchy of disjoint voter groups, charging each node's extra cost to its children's winners.
Each level adds at most $2$ to the distortion.

\subsection{Related Work}

\paragraph{Deterministic metric distortion.}
Distortion in social choice has been studied in the utilitarian setting~\cite{ProcacciaRosenschein2006,BoutilierEtAl2015}.
Anshelevich et al.~\cite{AnshelevichBhardwajPostl2015,AnshelevichBhardwajElkindPostlSkowron2018} introduced the metric distortion framework, proved that every deterministic voting rule has distortion at least $3$, and showed that Copeland achieves distortion $5$.
Munagala and Wang~\cite{MunagalaWang2019} improved the upper bound to $2+\sqrt5$ and identified a matching-based sufficient condition for distortion $3$.
Kempe~\cite{Kempe2020Duality} identified the same condition and developed an LP-duality framework that recovers the $2+\sqrt5$ bound and gives a simpler proof that the condition guarantees distortion $3$.
Gkatzelis et al.~\cite{GkatzelisHalpernShah2020} proved that every preference profile admits a candidate satisfying the stronger condition that its domination graph has a perfect matching, introduced the rule \rulename{PluralityMatching}, and settled the optimal distortion of deterministic rules at $3$.
Kizilkaya and Kempe~\cite{KizilkayaKempe2022} introduced \rulename{PluralityVeto}, a particularly simple voting rule that achieves the same optimal distortion.
Their later work connected the matching condition to the generalized veto core and introduced a continuous-time variant of \rulename{PluralityVeto}~\cite{KizilkayaKempe2023}.
For a comprehensive survey of distortion in social choice, see~\cite{AnshelevichFilosRatsikasShahVoudouris2021}.

\paragraph{Randomized metric distortion.}
For randomized voting rules, the optimal worst-case expected distortion remains open.
Random dictatorship, which elects the favorite candidate of a uniformly random voter, achieves expected distortion at most $3$~\cite{AnshelevichPostl2017,FeldmanFiatGolomb2016}.
Subsequent work improved the lower bound to above $2$~\cite{CharikarRamakrishnan2022,PulyassarySwamy2021} and the upper bound to below $3$~\cite{CharikarRamakrishnanWangWu2024,Frank2026,Ye2026,Shah2026}, with the current best bounds being approximately $2.1126$~\cite{CharikarRamakrishnan2022} and $2.1372$~\cite{CaiEtAl2026}, respectively.
Other work studies randomized rules with low sample complexity~\cite{GrossAnshelevichXia2017,FainGoelMunagalaPrabhu2019} and how additional samples reduce the risk of high distortion~\cite{FainFanMunagala2020}.
Whether near-optimal expected distortion can be achieved with low total communication and few rounds is an exciting avenue for future work.

\paragraph{Communication and rounds.}
Communication complexity in voting was studied by Conitzer and Sandholm~\cite{ConitzerSandholm2005}, who consider the number of bits required to compute standard voting rules exactly, and by Shah and Yu~\cite{ShahYu2026}, who study the elicitation of such rules through next-best queries.
In the metric distortion framework, Kempe~\cite{Kempe2020Communication} studies the distortion achievable by one-round mechanisms under communication constraints.
Ebadian et al.~\cite{EbadianHalpernMicha2024} studied mechanisms that elicit a bounded number of pairwise comparisons from each voter, under different restrictions on how queries may depend on earlier answers.
Prior work considered incomplete rankings or sampled voters~\cite{AnagnostidesFotakisPatsilinakos2022}.
In the utilitarian setting, prior work studied limited cardinal queries~\cite{AmanatidisBirmpasFilosRatsikasVoudouris2020,EbadianShah2025} and tradeoffs between communication and distortion~\cite{MandalProcacciaShahWoodruff2019,MandalShahWoodruff2020}.
Frank and Peters~\cite{FrankPeters2026} obtained distortion $2+\sqrt{5}$ from approval voting informed by plurality polls, which can be implemented in two rounds using $O(n\log m+nm)$ bits.
Our work focuses on deterministic protocols using only ordinal preferences, and on the tradeoff among distortion, round complexity, and total communication.

\paragraph{Voter aggregation and representation.}
Borodin et al.~\cite{BorodinLevShahStrangway2024} study primaries and tree-structured elections; their general distortion bound multiplies along tree paths and depends on participation fractions.
Work on distributed metric social choice~\cite{AnshelevichFilosRatsikasVoudouris2022,AbamEtAl2025} studied models in which voters are partitioned into districts, each district selects a representative alternative, and these alternatives are aggregated to select the final winner.
Prior work reduced participation by selecting a small representative panel~\cite{CaragiannisMichaPeters2024,EbadianMicha2025}, or allowed deliberation in small groups that could exploit collective knowledge of the underlying metric~\cite{GoelGoyalMunagala2025}.
Our protocol only uses ordinal information and filters candidates through a multilevel hierarchy, with every voter participating in each sub-election along their path to the root.

\section{Preliminaries}
\label{sec:prelim}

Throughout the paper, let $V$ be a set of $n$ voters and $C$ a set of $m$ candidates.
We assume that the voters and candidates have distinct public identifiers in $\{1,\ldots,n\}$ and $\{1,\ldots,m\}$, respectively.
Each voter $v \in V$ has a strict ranking $\succ_v$ over $C$, and we write $a \succeq_v b$ if $a \succ_v b$ or $a=b$.
We call $\profile=(\succ_v)_{v\in V}$ the \emph {preference profile}
and the triple $(V,C,\profile)$ an \emph{election}.

We frequently consider \emph{sub-elections} with a subset $U\subseteq V$ of voters and a subset $A\subseteq C$ of candidates, using only the rankings of voters in $U$ restricted to $A$.
A deterministic \emph{voting rule} $f$ returns a winner $f(U,A,\profile)\in A$ from such a sub-election.
When $\profile$ is clear from context, we write $f(U,A)$.
For every voter $v\in V$, let $\fav[A]{v}$ denote their favorite candidate in $A$, and let $\fav{v}:=\fav[C]{v}$ denote their global favorite candidate.

\paragraph{Metric distortion.}
We assume that voters and candidates lie in an unknown metric space with distance function $d$.
A metric $d$ is \emph{consistent} with $\profile$ if \(a\succeq_v b\) implies \(d(v,a)\le d(v,b)\) for every voter $v \in V$ and candidates $a,b \in C$.
Let $\mathcal{D}(\profile)$ denote the set of metrics consistent with $\profile$.
For a metric $d$, a subset of voters $U\subseteq V$, and a candidate $c\in C$, the \emph{social cost} of $c$ with respect to $U$ is \[ \cost[U]{c}:=\sum_{v\in U} d(v,c). \] The \emph{distortion} of a voting rule $f$ is the worst-case ratio between the social cost of the winner and the social cost of an optimal candidate:
\[ \sup_{\profile} \sup_{d\in\mathcal{D}(\profile)} \frac{\cost[V]{f(V,C, \profile)}} {\min_{o\in C}\cost[V]{o}}. \]

\paragraph{Communication model.}
We consider the blackboard communication model.
Communication proceeds in synchronous rounds.
In each round, every voter may simultaneously write a message on a public blackboard as a function of their ranking and the transcript from the preceding rounds.
Local computation is free.
The \emph{communication complexity} of a protocol is the total number of bits written, and its \emph{round complexity} is the number of rounds.

For a deterministic voting rule $f$ and a sub-election $(U,A)$, once
every voter in $U$ has reported their complete ranking over $A$, the
winner $f(U,A)$ can be determined from the public transcript without
further communication.

\paragraph{Plurality Veto and domination graphs.}
Our protocol uses Plurality Veto~\cite{KizilkayaKempe2022} as a subroutine.
We review its definition and the property of its winner that is relevant to our analysis.

\begin{definition}
[\rulename{PluralityVeto}~\cite{KizilkayaKempe2022}]
Given a set of voters $U$ and a set of candidates $A$, each candidate $c\in A$ starts with score $|\{v\in U:\fav[A]{v}=c\}|$.
The voters in $U$ are processed sequentially according to a veto order.
Each voter vetoes their least-preferred candidate among those whose current score is positive, decreasing that candidate's score by one.
The candidate whose score reaches zero last is the winner of $\rulename{PluralityVeto}(U,A)$.
\end{definition}

Plurality Veto is deterministic once its veto order is fixed.
Whenever it is applied to a sub-election $(U,A)$, the voters in $U$ veto in the order of their public identifiers.

For a candidate $w\in A$, its \emph{domination graph}~\cite{GkatzelisHalpernShah2020} with respect to $(U,A)$ is the bipartite graph $ G^{U,A}(w)=(U,U,E_w), $ where the two vertex sets are copies of $U$ and \[ (u,v)\in E_w \quad\Longleftrightarrow\quad w\succeq_u \fav[A]{v}. \]

For every veto order, the domination graph $G^{U,A}(w)$ of the winner $w=\rulename{PluralityVeto}(U,A)$ admits a perfect matching~\cite{KizilkayaKempe2022}.
To see this, let each voter $v\in U$ contribute one score unit labeled $v$ to their favorite candidate in $A$.
When voter $u$ casts a veto, remove any remaining score unit from the vetoed candidate and match $u$ to the voter $v$ labeling that unit.
Every voter casts exactly one veto, and every score unit is removed exactly once.
Moreover, $w$ has positive score before every veto, so \(w\succeq_u \fav[A]{v}\).
Hence every matched pair is an edge of $G^{U,A}(w)$, and they form a perfect matching.

\section{The \protocolname{} Protocol}
\label{sec:protocol}
Our voting protocol, \protocolname{}, is described in Algorithm~\ref{alg:tree-pv}.

In the first round, every voter reports their global favorite candidate.
Voters with the same reported favorite form a group, and these groups form the leaf nodes of a public rooted tree $T$.
We then select a winner at every node of $T$ in increasing order of height.\footnote{Every leaf node has height $0$, and the height of an internal node is one plus the maximum height of its children.}
At a leaf node $x$, the winner $w_x$ is the common global favorite of its voters.
At each internal node $x$, the available candidates $A_x$ consist of the distinct winners at the children of $x$, and the voters in the subtree rooted at $x$ run Plurality Veto on $A_x$ to select the winner $w_x$.
The winner at the root of $T$ is the final winner.
Figure~\ref{fig:tree-pv} illustrates this bottom-up selection process for $n=7$ voters and $r=3$ rounds.

\begin{algorithm}[ht]
\caption{\protocolname{} (\protocolshort)}
\label{alg:tree-pv}
\KwIn{A set $V$ of $n$ voters, a set $C$ of $m$ candidates,
and an integer $r\ge 2$}
\KwOut{A winning candidate $w\in C$}
Every voter $v\in V$ reports $\fav{v}$\;
$s \gets |\{ \fav{v} : v\in V \}|$\;
$q \gets \left\lceil s^{1/(r-1)} \right\rceil$\;
$H \gets \left\lceil \log_q s \right\rceil$\;
Partition $V$ into groups of voters with the same global favorite candidate\;
Fix a public rooted tree $T$ of height $H$ whose leaf nodes correspond to these groups and whose internal nodes have at most $q$ children\;
For every node $x\in T$, let $h_x$ denote its height and let $U_x\subseteq V$ be the voters in the subtree rooted at $x$\;
At every leaf node $x$, let $w_x$ be the common global favorite candidate of the voters in $U_x$\;
\For{$t\gets 1$ \KwTo $H$}{
    \ForEach{node $x$ with $h_x=t$ in parallel}{
        $A_x\gets\bigcup_{y:\,y\text{ is a child of }x}\{w_y\}$\;
        Every voter in $U_x$ reports their complete ranking over $A_x$\;
        $w_x\gets\rulename{PluralityVeto}(U_x,A_x)$\;
    }
}
\Return the winner at the root of $T$\;
\end{algorithm}

\begin{figure}[ht]
\centering
\begin{tikzpicture}[
    >=Latex,
    root/.style={draw, rounded corners, align=center, font=\small, inner sep=3pt},
    election/.style={draw, rounded corners, align=center, font=\small, inner sep=3pt},
    voter/.style={draw, rounded corners, align=center, font=\small, inner sep=3pt},
    flow/.style={->,semithick}
]
\def\RootX{-2.32}
\def\RootY{0}
\def\MiddleY{-2.5}
\def\LeafY{-5}
\def\XOneX{-7.3}
\def\XTwoX{0}
\def\VSevenX{4.3}
\def\LeafGap{2.4}

\node[root, anchor=north] (x3) at (\RootX,\RootY) {
    root $x_3$\\ $U_{x_3}=V$\\
    $A_{x_3}=\{w_{x_1},w_{x_2},w_{v_7}\}$\\
    $w_{x_3}\gets\rulename{PluralityVeto}(U_{x_3},A_{x_3})$
};
\node[election, anchor=north] (x1) at (\XOneX,\MiddleY) {
    node $x_1$\\ $U_{x_1}=\{v_1,v_2,v_3\}$\\
    $A_{x_1}=\{w_{v_1},w_{v_2},w_{v_3}\}$\\
    $w_{x_1}\gets\rulename{PluralityVeto}(U_{x_1},A_{x_1})$
};
\node[election, anchor=north] (x2) at (\XTwoX,\MiddleY) {
    node $x_2$\\ $U_{x_2}=\{v_4,v_5,v_6\}$\\
    $A_{x_2}=\{w_{v_4},w_{v_5},w_{v_6}\}$\\
    $w_{x_2}\gets\rulename{PluralityVeto}(U_{x_2},A_{x_2})$
};
\node[voter, anchor=north] (v7) at (\VSevenX,\MiddleY) {
    node $v_7$\\ $w_{v_7}=\fav{v_7}$
};

\node[voter, anchor=north] (v1) at ({\XOneX-\LeafGap},\LeafY) {node $v_1$\\ $w_{v_1}=\fav{v_1}$};
\node[voter, anchor=north] (v2) at (\XOneX,\LeafY) {node $v_2$\\ $w_{v_2}=\fav{v_2}$};
\node[voter, anchor=north] (v3) at ({\XOneX+\LeafGap},\LeafY) {node $v_3$\\ $w_{v_3}=\fav{v_3}$};
\node[voter, anchor=north] (v4) at ({\XTwoX-\LeafGap},\LeafY) {node $v_4$\\ $w_{v_4}=\fav{v_4}$};
\node[voter, anchor=north] (v5) at (\XTwoX,\LeafY) {node $v_5$\\ $w_{v_5}=\fav{v_5}$};
\node[voter, anchor=north] (v6) at ({\XTwoX+\LeafGap},\LeafY) {node $v_6$\\ $w_{v_6}=\fav{v_6}$};

\draw[flow] (v1.north) -- ($(x1.south west)!0.25!(x1.south east)$);
\draw[flow] (v2.north) -- ($(x1.south west)!0.5!(x1.south east)$);
\draw[flow] (v3.north) -- ($(x1.south west)!0.75!(x1.south east)$);
\draw[flow] (v4.north) -- ($(x2.south west)!0.25!(x2.south east)$);
\draw[flow] (v5.north) -- ($(x2.south west)!0.5!(x2.south east)$);
\draw[flow] (v6.north) -- ($(x2.south west)!0.75!(x2.south east)$);
\draw[flow] (x1.north) -- ($(x3.south west)!0.25!(x3.south east)$);
\draw[flow] (x2.north) -- ($(x3.south west)!0.5!(x3.south east)$);
\draw[flow] (v7.north) -- ($(x3.south west)!0.75!(x3.south east)$);
\end{tikzpicture}
\caption{An illustration of \protocolname{} with $n=7$ and $r=3$, where all voters have distinct global favorite candidates.}
\label{fig:tree-pv}
\end{figure}
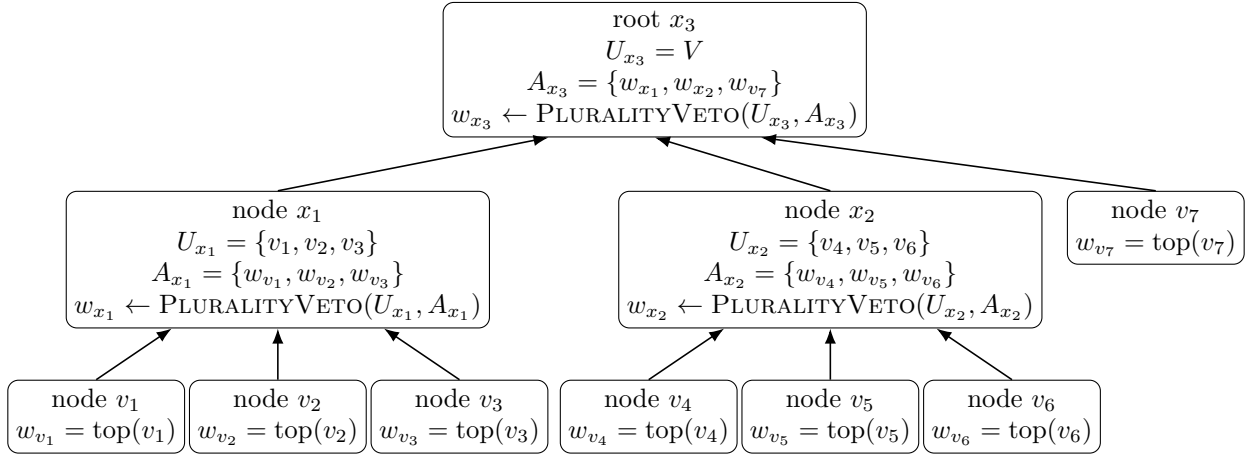

\begin{remark}[Voting rules at internal nodes]
\label{rem:internal-voting-rule}
Algorithm~\ref{alg:tree-pv} can use any deterministic voting rule $f$ at its internal nodes.
For concreteness, we state it with Plurality Veto, which suffices to obtain our main result.
The distortion analysis in Section~\ref{subsec:analysis-distortion} requires only that, for every sub-election $(U,A)$, the domination graph of the winner $f(U,A)$ admits a perfect matching.
Plurality Veto satisfies this property because its veto process induces such a matching.
The communication analysis in Section~\ref{subsec:analysis-communication} applies to any deterministic voting rule $f$.
\end{remark}

\section{Analysis of the Protocol}
\label{sec:analysis}

We upper bound the distortion of \protocolname{} (Lemma~\ref{lem:distortion}) in Section~\ref{subsec:analysis-distortion} and analyze its round and communication complexity (Lemma~\ref{lem:communication}) in Section~\ref{subsec:analysis-communication}.
Theorem~\ref{thm:main} follows from these two lemmas, since there are at most $\min\{n,m\}$ distinct global favorite candidates.

\subsection{Metric Distortion}
\label{subsec:analysis-distortion}

\begin{lemma}
\label{lem:distortion}
For every node $x$ of $T$ and every candidate $o \in C$,
\[
\cost[U_x]{w_x} \le (2h_x + 1) \cost[U_x]{o}.
\]
Consequently, \protocolname{} (Algorithm~\ref{alg:tree-pv}) has distortion at most $2r - 1$.
\end{lemma}

\begin{proof}
We prove the lemma for every node $x$ by induction on $h_x$ from the leaf nodes to the root.

For the base case, let $x$ be a leaf node.
Then $h_x = 0$, and $w_x$ is the global favorite candidate of every voter in $U_x$.
For every candidate $o \in C$, we have $\cost[U_x]{w_x} = \sum_{v\in U_x}d(v,w_x) \le \sum_{v\in U_x}d(v,o) =\cost[U_x]{o}$.

For the inductive step, let $x$ be an internal node, and suppose that the inequality in the lemma holds for every child node of $x$.
Let $U := U_x$, $A := A_x$, and $w := w_x$.
Fix any candidate $o \in C$.
Because $w$ is the Plurality Veto winner, its domination graph $G^{U,A}(w)$ admits a perfect matching.
Let $\phi: U \to U$ be the bijection induced by this matching.
We have $w \succeq_v \fav[A]{\phi(v)}$ for every $v \in U$, and by the triangle inequality,
\[
d(v, w) \le d\bigl(v, \fav[A]{\phi(v)}\bigr) \le d(v, o) + d\bigl(o, \phi(v)\bigr) + d\bigl(\phi(v), \fav[A]{\phi(v)}\bigr).
\]
Summing over $v \in U$ and using that $\phi$ is a bijection,
\begin{equation}
\label{eq:matching-inequality}
\cost[U]{w} \le 2\cost[U]{o} + \sum_{v\in U} d\bigl(v, \fav[A]{v}\bigr).
\end{equation}

For each child node $y$ of $x$, we have $w_y \in A$, so every $v \in U_y$ satisfies
$d\bigl(v, \fav[A]{v}\bigr)\le d(v, w_y)$.
Summing these inequalities over all child nodes $y$ of $x$ and all
voters $v\in U_y$ gives
\begin{align*}
\sum_{v\in U} d\bigl(v, \fav[A]{v}\bigr)
&\le \sum_{y:\,y\text{ is a child of }x} \cost[U_y]{w_y}\\
&\le \sum_{y:\,y\text{ is a child of }x} (2h_y + 1) \cost[U_y]{o}\\
&\le (2h_x - 1) \sum_{y:\,y\text{ is a child of }x} \cost[U_y]{o}
= (2h_x - 1)\cost[U]{o}.
\end{align*}
The second inequality uses the inductive hypothesis, and the third inequality uses $h_y \le h_x-1$.
Combining this with Inequality~\eqref{eq:matching-inequality} proves $\cost[U_x]{w_x} \le (2h_x + 1)\cost[U_x]{o}$.

Finally, $q=\left\lceil s^{1/(r-1)}\right\rceil$ implies
$H=\left\lceil\log_q s\right\rceil\le r-1$.
Let $w$ be the final winner of Algorithm~\ref{alg:tree-pv}.
Applying the inequality in the lemma at the root gives, for every $o \in C$,
\[
\cost[V]{w} \le (2H + 1)\cost[V]{o} \le (2r - 1)\cost[V]{o}. \qedhere
\]
\end{proof}

As noted in Remark~\ref{rem:internal-voting-rule}, the only property of Plurality Veto used in the proof of Lemma~\ref{lem:distortion} is that the domination graph of its winner admits a perfect matching.
Consequently, Lemma~\ref{lem:distortion} applies to any voting rule satisfying this property.
In Appendix~\ref{apx:rules-without-matching}, we show that Algorithm~\ref{alg:tree-pv} can also use voting rules without this property, although the resulting distortion bound may be exponential rather than linear in $r$.

\subsection{Round and Communication Complexity}
\label{subsec:analysis-communication}

\begin{lemma}
\label{lem:communication}
Let $n$ be the number of voters, $m$ the number of candidates, and $s$ the number of distinct global favorite candidates.
\protocolname{} (Algorithm~\ref{alg:tree-pv}) uses at most $r$ rounds and $O\!\left(n\log m+n s^{1/(r-1)}\log s\right)$ bits of total communication.
\end{lemma}

\begin{proof}
In the first round, every voter reports their global favorite candidate using $O(\log m)$ bits, for a total of $O(n\log m)$ bits.
Then the voter groups and the tree $T$ can be constructed canonically from the public transcript and $r$, so no further communication is required to specify them.

Fix $h\ge 1$ and consider all internal nodes $x$ of height $h$.
For each such node, the candidate set $A_x$ can be determined from the public transcript and contains at most $q$ candidates because $x$ has at most $q$ child nodes.
The candidates in $A_x$ can be sorted by their public identifiers and re-indexed by $\{1, \ldots, |A_x|\}$.
Every voter in $U_x$ can therefore report their complete ranking over $A_x$ using $O(|A_x|\log |A_x|)=O(q\log q)$ bits.
The voter sets $U_x$ associated with nodes of height $h$ are disjoint.
Consequently, all rankings required for the sub-elections at nodes of height $h$ can be collected in one round using $O(nq\log q)$ bits of communication.
Since the voting rule used at each internal node is deterministic, the winners at these nodes can then be determined from the public transcript without further communication.

The protocol uses one round for the voters to report their global favorite candidates, and then one round for each level at height $h\in\{1,\ldots,H\}$.
The choice of $q=\left\lceil s^{1/(r-1)}\right\rceil$ implies $H=\left\lceil\log_q s\right\rceil\le r-1$.
The total number of rounds is therefore $H+1\le r$.

The total communication is $O\bigl(n\log m+Hnq\log q\bigr)$ bits.
Since $H\log q=O(\log s)$ and $q=O\!\left(s^{1/(r-1)}\right)$, this yields the claimed bound of $O\!\left(n\log m+n s^{1/(r-1)}\log s\right)$ bits.
\end{proof}

\section{Acknowledgment}
We thank David Kempe for discussions during the early stages of this project. ChatGPT, Claude, and Gemini were used to assist in exploring possible approaches, conducting literature searches, and polishing the text. All proofs in this paper were written by the authors. The authors assume responsibility for all content.

\bibliographystyle{alpha}
\bibliography{references}

\appendix

\section{A Lower Bound on Total Communication}
\label{apx:lower-bound}
We show that constant distortion requires $\Omega(n)$ bits of total communication even with an unlimited number of rounds.
If a deterministic protocol communicates at most $n/\beta$ bits for some $\beta>1$, then there is an execution in which the preferences between two candidates remain undetermined for at least a $(1-1/\beta)$ fraction of voters.
All of these voters can prefer the candidate that is not elected.

Similar constructions have appeared in lower bounds for metric distortion (e.g.,~\cite{AnshelevichBhardwajPostl2015,AnagnostidesFotakisPatsilinakos2022,EbadianHalpernMicha2024}).
We formalize the argument for the blackboard communication model below.

\begin{lemma}
\label{lem:lower-bound}
For $m\ge 2$, every deterministic voting protocol with distortion at most $\alpha\ge 3$ requires at least $\frac{2n}{\alpha+1}$ bits of total communication in the worst case, regardless of the number of rounds.
\end{lemma}
\begin{proof}
Consider a line metric space.
Fix two candidates $w$ and $o$.
Consider only profiles in which every voter ranks $w$ and $o$ above all other candidates, with all other candidates placed arbitrarily far away from the voters in a fixed order.
Therefore, each voter has only two possible rankings.

We construct a transcript while leaving some voters' rankings undetermined.
In each round, each voter whose ranking remains undetermined has two possible messages to send (which could be empty).
If the two messages coincide, we leave the ranking undetermined.
Otherwise, fix the ranking to one that produces a nonempty message.
Voters whose rankings are fixed follow the protocol.
At termination, let $S$ be the set of voters whose rankings were fixed.
Each voter in $S$ sends at least one bit when their ranking is fixed, so the protocol communicates at least $|S|$ bits.

Every completion of the undetermined rankings produces the same transcript and hence the same winner.
The winner must be in $\{w,o\}$, otherwise the distortion is unbounded.
Without loss of generality, suppose that $w$ is elected.
Complete every undetermined ranking with $o\succ w$.

We can place $w$ at $-1$, $o$ at $1$, each voter in $S$ at $-\eps$ or $\eps$ for arbitrarily small $\eps > 0$, and all remaining voters at $1$.
This is consistent with the completed profile.
Consequently, as $\eps\to0$, the distortion guarantee implies $\cost[V]{w} = 2n - |S| \le \alpha |S| = \alpha \cost[V]{o}$, so $|S|\ge \frac{2n}{\alpha+1}$.
\end{proof}

\section{Voting Rules Without the Perfect-Matching Property}
\label{apx:rules-without-matching}

We show that a weaker distortion bound holds even when the voting rule used at the internal nodes does not satisfy the perfect-matching property.

\begin{lemma}
\label{lem:other-rules}
Suppose that Plurality Veto in Algorithm~\ref{alg:tree-pv} is replaced by a deterministic voting rule $f$ whose distortion is at most $\alpha \ge 3$.
The resulting protocol has distortion $O\!\left(\alpha^{r-1}\right)$.
\end{lemma}
\begin{proof}
Let $D_0 := 1$ and $D_h := \alpha(D_{h-1}+2)$ for every $h \ge 1$.
Following the same leaf-to-root induction as in the proof of Lemma~\ref{lem:distortion}, we show that
\[
\cost[U_x]{w_x}\le D_{h_x}\cost[U_x]{o}
\]
for every node $x$ and every candidate $o\in C$.
The base case holds because $D_0=1$ and the winner at a leaf node is the global favorite candidate of every voter in $U_x$.

Consider an internal node $x$.
Let $h := h_x$ and $U := U_x$.
Fix any $o\in C$.
By averaging, some child node $y$ of $x$ satisfies $\frac{\cost[U_y]{o}}{|U_y|} \le \frac{\cost[U]{o}}{|U|}$.
By the triangle inequality,
\begin{align*}
\cost[U]{w_y}
&\le \cost[U]{o} + |U| \, d(o,w_y)\\
&\le \cost[U]{o} + \frac{|U|}{|U_y|} \!\left(\cost[U_y]{o} + \cost[U_y]{w_y}\right)\\
&\le \cost[U]{o} + (D_{h-1} + 1) \frac{|U|}{|U_y|}\cost[U_y]{o}\\
&\le (D_{h-1} + 2) \cost[U]{o}.
\end{align*}
The third inequality uses the inductive hypothesis,
$h_y\le h-1$, and that $D_h$ is increasing in $h$.
The last inequality uses the choice of $y$.

Because $w_y\in A_x$ and $f$ has distortion at most $\alpha$,
\[
\cost[U]{w_x} \le \alpha\min_{a\in A_x}\cost[U]{a} \le \alpha\cost[U]{w_y} \le \alpha(D_{h-1}+2)\cost[U]{o} = D_h\cost[U]{o}.
\]

Solving the recurrence gives
\[
D_h
=
\alpha^h+\frac{2\alpha(\alpha^h-1)}{\alpha-1}
\le
4\alpha^h,
\]
where the inequality uses $\alpha\ge3$.
Since $H\le r-1$, the final distortion is at most $D_H=O\!\left(\alpha^{r-1}\right)$.
\end{proof}

\end{document}